\documentclass[journal,onecolumn]{IEEEtran}
\usepackage{amsmath}
\usepackage{amssymb}
\usepackage{mathrsfs}
\usepackage{cite}
\usepackage{epsfig}
\usepackage{epsf}
\usepackage{theorem}
\usepackage{graphics}
\usepackage{hyperref}

\newtheorem{thm}{Theorem}
\newtheorem{lem}{Lemma}

\newtheorem{defi}{Definition}

\newtheorem{rem}{Remark}

\begin{document}

%opening
\title{The Secrecy Rate Region of the Broadcast Channel}

\author{Ghadamali Bagherikaram, Abolfazl S. Motahari, Amir K. Khandani\\
Coding and Signal Transmission Laboratory,\\
 Department of
Electrical
and Computer Engineering,\\
 University of Waterloo, Waterloo, Ontario,
 N2L 3G1\\
 Emails: \{gbagheri,abolfazl,khandani\}@cst.uwaterloo.ca}

\maketitle
\begin{abstract}
In this paper, we consider a scenario where a source node wishes to
broadcast two confidential messages for two respective receivers,
while a wire-tapper also receives the transmitted signal. This model
is motivated by wireless communications, where individual secure
messages are broadcast over open media and can be received by any
illegitimate receiver. The secrecy level is measured by equivocation
rate at the eavesdropper. We first study the general (non-degraded)
broadcast channel with confidential messages. We present an inner
bound on the secrecy capacity region for this model. The inner bound
coding scheme is based on a combination of random binning and the
Gelfand-Pinsker bining. This scheme matches the Marton's inner bound
on the broadcast channel without confidentiality constraint. We
further study the situation where the channels are degraded. For the
degraded broadcast channel with confidential messages, we present
the secrecy capacity region. Our achievable coding scheme is based
on Cover's superposition scheme and random binning. We refer to this
scheme as Secret Superposition Scheme. In this scheme, we show that
randomization in the first layer increases the secrecy rate of the
second layer. This capacity region matches the capacity region of
the degraded broadcast channel without security constraint. It also
matches the secrecy capacity for the conventional wire-tap channel.
Our converse proof is based on a combination of the converse proof
of the conventional degraded broadcast channel and Csiszar lemma.
Finally, we assume that the channels are Additive White Gaussian
Noise (AWGN) and show that secret superposition scheme with Gaussian
codebook is optimal. The converse proof is based on the generalized
entropy power inequality.
\end{abstract}
\section{Introduction}
The notion of information theoretic secrecy in communication systems
was first introduced by Shannon in \cite{1}. The information
theoretic secrecy requires that the received signal of the
eavesdropper does not provide even a single bit information about
the transmitted messages. Shannon considered a pessimistic situation
where both the intended receiver and the eavesdropper have direct
access to the transmitted signal (which is called ciphertext). Under
these circumstances, he proved a negative result showing that
perfect secrecy can be achieved only when the entropy of the secret
key is greater than or equal to the entropy of the message. In
modern cryptography, all practical cryptosystems are based on
Shannnon's pessimistic assumption. Due to practical constraints,
secret keys are much shorter than messages. Therefore, these
practical cryptosystems are theoretically susceptible of breaking by
attackers. However, the goal of designing such practical ciphers is
to guarantee that there exists no efficient algorithm for breaking
them.

Wyner in \cite{2} showed that the above negative result is a
consequence of Shannon's restrictive assumption that the adversary
has access to precisely the same information as the legitimate
receiver. Wyner considered a scenario in which a wire-tapper
receives the transmitted signal over a degraded channel with respect
to the legitimate receiver's channel. He further assumed that the
wire-tapper has no computational limitations and knows the codebook
used by the transmitter. He measured the level of ignorance at the
eavesdropper by its equivocation and characterized the
capacity-equivocation region. Interestingly, a non-negative perfect
secrecy capacity is always achievable for this scenario.

The secrecy capacity for the Gaussian wire-tap channel is
characterized by Leung-Yan-Cheong in \cite{3}. Wyner's work then is
extended to the general (non-degraded) broadcast channel with
confidential messages (BCC) by Csiszar and Korner \cite{4}. They
considered transmitting confidential information to the legitimate
receiver while transmitting common information to both the
legitimate receiver and the wire-tapper. They established a
capacity-equivocation region of this channel. %In this model, when
%the main channel is less noisy or more capable than
%the eavesdropper's channel, it is possible to achieve a non-zero secrecy capacity.\\

The BCC is further studied recently in \cite{5,6,7}, where the
source node transmits a common message for both receivers, along
with two additional confidential messages for two respective
receivers. The fading BCC is investigated in \cite{8,9} where the
broadcast channels from the source node to the legitimate receiver
and the eavesdropper is corrupted by multiplicative fading gain
coefficients, in addition to additive white Gaussian noise terms.
The Channel State Information (CSI) is assumed to be known at the
transmitter. In \cite{10}, the perfect secrecy capacity is derived
where the channels are slow fading. Moreover, the optimal power
control policy is obtained for different scenarios regarding
availability of CSI. In \cite{11}, the wire-tap channel is extended
to the parallel broadcast channels and the fading channels with
multiple receivers. Here, the secrecy constraint is a perfect
equivocation for each messages, even if all the other messages are
revealed to the eavesdropper. The secrecy sum capacity for a reverse
broadcast channel is derived for this restrictive assumption. The
notion of the wire-tap channel is also extended to multiple access
channels \cite{12,13,14,15}, relay channels \cite{16,17,18,19},
parallel channels \cite{20} and MIMO channels
\cite{21,22,23,24,25,26}. Some other related works on communication
of confidential messages can be found in \cite{27,28,29,30,31}.

In this paper, we consider a scenario where a source node wishes to
broadcast two confidential messages for two respective receivers,
while a wire-tapper also receives the transmitted signal. This model
is motivated by wireless communications, where individual secure
messages are broadcast over shared media and can be received by any
illegitimate receiver. In fact, we simplify the restrictive
constraint imposed in \cite{11} and assume that the eavesdropper
does not have access to the other messages. We first study the
general broadcast channel with confidential messages. We present an
achievable rate region for this channel. Our achievable coding
scheme is based on the combination of the random binning and the
Gelfand-Pinsker bining \cite{32}. This scheme matches the Marton's
inner bound \cite{33} on the broadcast channel without
confidentiality constraint. We further study the situation where the
channels are physically degraded and characterize the secrecy
capacity region. Our achievable coding scheme is based on Cover's
superposition coding \cite{34} and the random binning. We refer to
this scheme as Secret Superposition Coding. This capacity region
matches the capacity region of the degraded broadcast channel
without security constraint. It also matches the secrecy capacity of
the wire-tap channel.

The rest of the paper is organized as follows. In section II we
introduce the system model. In Section III, we provide an inner
bound on the secrecy capacity region when the channels are not
degraded. In section IV, we specialize our channel to the physically
degraded and establish the secrecy capacity region. In Section V, we
conclude the paper.

\section{Preliminaries}

In this paper, a random variable is denoted by a capital letter
(e.g. X) and its realization is denoted by a corresponding lower
case letter (e.g. x). The finite alphabet of a random variable is
denoted by a script letter (e.g. $\mathcal{X}$) and its probability
distribution is denoted by $P(x)$. Let  $\mathcal{X}$ be a finite
alphabet set and denote its cardinality by $|\mathcal{X}|$. The
members of  $\mathcal{X}^{n}$ will be written as
$x^{n}=(x_{1},x_{2},...,x_{n})$, where subscripted letters denote
the components and superscripted letters denote the vector. The
notation $x^{i-1}$ denotes the vector $(x_{1},x_{2},...,x_{i-1})$. A
similar notation will be used for random variables and random
vectors.

Consider a Broadcast Channel with Confidential messages as depicted
in fig.\ref{f1}.
\begin{figure}
\centerline{\includegraphics[scale=.6]{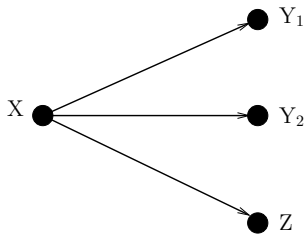}} \caption{Broadcast
Channel with Confidential Messages} \label{f1}
\end{figure}
 In this confidential setting, the transmitter ($X$) wants to broadcast some secret messages to the
legitimated receivers ($Y_{1}$,$Y_{2}$), and prevent the
eavesdropper ($Z$) from having any information about the messages. A
discrete memoryless broadcast channel with confidential messages is
described by finite sets $\mathcal{X}$,
$\mathcal{Y}_{1}$,$\mathcal{Y}_{2}$,$\mathcal{Z}$, and a conditional
distribution $P(y_{1},y_{2},z|x)$. The input of the channel is
$x\in\mathcal{X}$ and the outputs are
$(y_{1},y_{2},z)\in(\mathcal{Y}_{1}\times\mathcal{Y}_{2}\times\mathcal{Z})$
for receiver $1$, receiver $2$, and the eavesdropper, respectively.
The transmitter wishes to send independent messages $(W_{1},W_{2})$
to the respective receivers in $n$ uses of the channel while
insuring perfect secrecy. The channel is discrete memoryless in the
sense that
\begin{equation}
P(y_{1}^{n},y_{2}^{n},z^{n}|x^{n})=\prod_{i=1}^{n}P(y_{1,i},y_{2,i},z_{i}|x_{i}).
\end{equation}
A $((2^{nR_{1}},2^{nR_{2}}),n)$ code for a broadcast channel with
confidential messages consists of a stochastic encoder
\begin{equation}
f:(\{1,2,...,2^{nR_{1}}\}\times\{1,2,...,2^{nR_{2}}\})\rightarrow
\mathcal{X}^{n},
\end{equation}
and two decoders,
\begin{equation}
g_{1}:\mathcal{Y}_{1}^{n}\rightarrow \{1,2,...,2^{nR_{1}}\}
\end{equation}
and
\begin{equation}
g_{2}:\mathcal{Y}_{2}^{n}\rightarrow \{1,2,...,2^{nR_{2}}\}.
\end{equation}
The average probability of error is defined as the probability that
the decoded messages are not equal to the transmitted messages; that
is,
\begin{equation}
P_{e}^{(n)}=P(g_{1}(Y_{1}^{n})\neq W_{1}\cup g_{2}(Y_{2}^{n})\neq
W_{2}).
\end{equation}

The knowledge that the eavesdropper gets about $W_{1}$ and $W_{2}$
from its received signal $Z^{n}$ is modeled as
\begin{IEEEeqnarray}{lr}
I(Z^{n},W_{1})=H(W_{1})-H(W_{1}|Z^{n}),\\
I(Z^{n},W_{2})=H(W_{2})-H(W_{2}|Z^{n}),
\end{IEEEeqnarray}
and
\begin{equation}
I(Z^{n},(W_{1},W_{2}))=H(W_{1},W_{2})-H(W_{1},W_{2}|Z^{n}).
\end{equation}
Perfect secrecy revolves around the idea that the eavesdropper
cannot get even a single bit information about the transmitted
messages. Perfect secrecy thus requires that
\begin{IEEEeqnarray}{lr}
I(Z^{n},W_{1})=0\Leftrightarrow H(W_{1})=H(W_{1}|Z^{n}),\\
\nonumber I(Z^{n},W_{2})=0\Leftrightarrow H(W_{2})=H(W_{2}|Z^{n}),
\end{IEEEeqnarray}
and
\begin{equation}
I(Z^{n},(W_{1},W_{2}))=0\Leftrightarrow
H(W_{1},W_{2})=H(W_{1},W_{2}|Z^{n}).
\end{equation}
 The secrecy levels of confidential messages $W_{1}$ and
$W_{2}$ are measured at the eavesdropper in terms of equivocation
rates which are defined as follows.
\begin{defi}
The equivocation rates $R_{e1}$, $R_{e2}$ and $R_{e12}$
 for the
Broadcast channel with confidential messages are:
\begin{IEEEeqnarray}{lr}
R_{e1}=\frac{1}{n}H(W_{1}|Z^{n}),\\
\nonumber R_{e2}=\frac{1}{n}H(W_{2}|Z^{n}), \\ \nonumber
R_{e12}=\frac{1}{n}H(W_{1},W_{2}|Z^{n}).
\end{IEEEeqnarray}
\end{defi}
The perfect secrecy rates $R_{1}$ and $R_{2}$ are the amount of
information that can be sent to the legitimate receivers not only
reliably but also confidentially.
\begin{defi}
A secrecy rate pair $(R_{1},R_{2})$ is said to be achievable if for
any $\epsilon>0$, there exists  a sequence of
$((2^{nR_{1}},2^{nR_{2}}),n)$ codes, such that for sufficiently
large $n$, we have:
\begin{IEEEeqnarray}{rl}
\label{l0}P_{e}^{(n)}&\leq \epsilon,\\
\label{l1}
 R_{e1}&\geq R_{1}-\epsilon_{1},\\
\label{l2} R_{e2}&\geq R_{2}-\epsilon_{2},\\
\label{l3}
 R_{e12}&\geq R_{1}+R_{2}-\epsilon_{3}.
%\frac{1}{n}H(W_{1},W_{2}|Z) \geq \frac{1}{n}H(W_{1},W_{2})-\epsilon
\end{IEEEeqnarray}
\end{defi}
In the above definition, the first condition concerns the
reliability, while the other conditions guarantee perfect secrecy for each individual message and both messages as well.
%\begin{rem}
%Note that if both messages satisfy perfect secrecy constraint of
%(\ref{l3}), then both constraints of (\ref{l1}) or (\ref{l2}) will
%be satisfied. Let us assume (\ref{l3}) is satisfied. Then for
%message $W_{1}$ we have
%\begin{IEEEeqnarray}{lr}\nonumber
%\frac{1}{n}H(W_{1},W_{2}|Z) \geq \frac{1}{n}H(W_{1},W_{2})-\epsilon_{3},\\
%\nonumber \frac{1}{n}H(W_{1}|Z) +\frac{1}{n}H(W_{2}|W_{1},Z) \geq
%\frac{1}{n}H(W_{1})+\frac{1}{n}H(W_{2}|W_{1})-\epsilon_{3},\\
%\nonumber  \frac{1}{n}H(W_{1}|Z) \geq
%\frac{1}{n}H(W_{1})+\frac{1}{n}H(W_{2}|W_{1})-\frac{1}{n}H(W_{2}|W_{1},Z)-\epsilon_{3},\\
%\nonumber
% \frac{1}{n}H(W_{1}|Z) \geq \frac{1}{n}H(W_{1})-\epsilon_{3}.
%\end{IEEEeqnarray}
%where we used the fact that
%$\frac{1}{n}H(W_{2}|W_{1})-\frac{1}{n}H(W_{2}|W_{1},Z)\geq 0$. The
%same argument can be used for $H(W_{2}|Z)$.
%\end{rem}
 The capacity region is defined as follows.
\begin{defi}
The capacity region of the broadcast channel with confidential
messages is the closure of the set of all achievable rate pairs
$(R_{1},R_{2})$.
\end{defi}

\section{General BCCs}
In this section, we consider the general broadcast channel with
confidential messages and present an achievable rate region. Our
achievable coding scheme is based on a combination of the random
binning and the Gelfand-Pinsker bining schemes \cite{32}. The
following theorem illustrates the achievable rate region for this
channel.
%\subsection{Inner Bound for General BCCs}
\begin{thm}
Let $\mathbb{R}_{I}$ denote the union of all non-negative rate pairs
$(R_{1},R_{2})$ satisfying
 \begin{IEEEeqnarray}{rl}
    R_{1}&\leq I(V_{1};Y_{1})-I(V_{1};Z), \\
    R_{2}&\leq I(V_{2};Y_{2})-I(V_{2};Z), \\
    R_{1}+R_{2}&\leq
    I(V_{1};Y_{1})+I(V_{2};Y_{2})-I(V_{1},V_{2};Z)-I(V_{1};V_{2}).
  \end{IEEEeqnarray}
over all joint distributions
$P(v_{1},v_{2})P(x|v_{1},v_{2})P(y_{1},y_{2},z|x)$. Then any rate
pair $(R_{1},R_{2})\in \mathbb{R}_{I}$ is achievable for the
broadcast channel with confidential messages.
\end{thm}

\begin{rem}
If we remove the secrecy constraints by setting
$\mathcal{Z}=\emptyset$, then the above rate region reduces to
Marton's achievable region for the general broadcast channel.
\end{rem}
\begin{rem}
If we remove one of the users by setting e.g.,
$\mathcal{Y}_{2}=\emptyset$, then we get the Csiszar and Korner's
secrecy capacity for the other user.
\end{rem}

\begin{proof}

1) \textit{Codebook Generation}:
 The structure of the encoder is
depicted in Fig.\ref{f2}.
\begin{figure}
\centerline{\includegraphics[scale=.6]{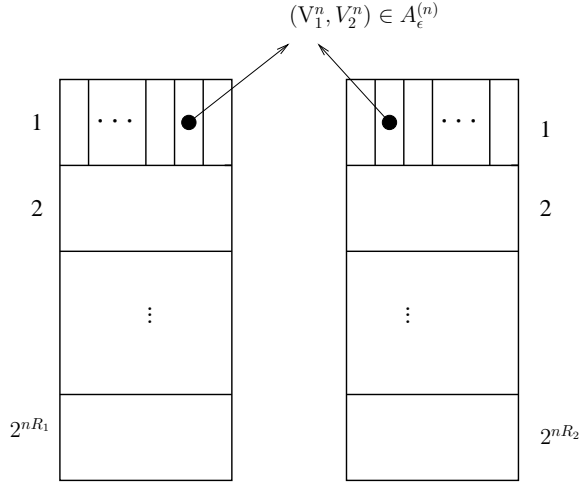}} \caption{The
Stochastic Encoder} \label{f2}
\end{figure}
Fix $P(v_{1})$, $P(v_{2})$ and  $P(x|v_{1},v_{2})$. The stochastic
encoder generates $2^{n(I(V_{1};Y_{1})-\epsilon)}$ independent and
identically distributed sequences $v_{1}^{n}$ according to the
distribution $P(v_{1}^{n})=\prod_{i=1}^{n}P(v_{1,i})$. Next,
randomly distribute these sequences into $2^{nR_{1}}$ bins such that
each bin contains $2^{n(I(V_{1};Z)-\epsilon)}$ codewords. Similarly,
it generates $2^{n(I(V_{2};Y_{2})-\epsilon)}$ independent and
identically distributed sequences $v_{2}^{n}$ according to the
distribution $P(v_{2}^{n})=\prod_{i=1}^{n}P(v_{2,i})$. Next,
randomly distribute these sequences into $2^{nR_{2}}$ bins such that
each bin contains $2^{n(I(V_{2};Z)-\epsilon)}$ codewords. Index each
of the above bins by $w_{1}\in\{1,2,...,2^{nR_{1}}\}$ and
$w_{2}\in\{1,2,...,2^{nR_{2}}\}$ respectively.

2) \textit{Encoding}: To send messages $w_{1}$ and $w_{2}$, the
transmitter looks for $v_{1}^{n}$ in bin $w_{1}$ of the first bin
set and looks for $v_{2}^{n}$ in bin $w_{2}$ of the second bin set,
such that $(v_{1}^{n},v_{2}^{n})\in
A_{\epsilon}^{(n)}(P_{V_{1},V_{2}})$ where
$A_{\epsilon}^{(n)}(P_{V_{1},V_{2}})$ denotes the set of jointly
typical sequences $v_{1}^{n}$ and $v_{2}^{n}$ with respect to
$P(v_{1},v_{2})$. The rates are such that there exist more than one
joint typical pair, the transmitter randomly chooses one of them and
then generates $x^{n}$ according to
$P(x^{n}|v_{1}^{n},v_{2}^{n})=\prod_{i=1}^{n}P(x_{i}|v_{1,i},v_{2,i})$.
This scheme is equivalent to the scenario in which each bin is
divided into subbins and the transmitter randomly chooses one of the
subbins of bin $w_{1}$ and one of the subbins of bin $w_{2}$. It
then looks for a joint typical sequence $(v_{1}^{n},v_{2}^{n})$ in
the corresponding subbins and generates $x^{n}$.

3) \textit{Decoding}: The received signals at the legitimate
receivers, $y_{1}^{n}$ and $y_{2}^{n}$, are the outputs of the
channels $P(y_{1}^{n}|x^{n})=\prod_{i=1}^{n}P(y_{1,i}|x_{i})$ and
$P(y_{2}^{n}|x^{n})=\prod_{i=1}^{n}P(y_{2,i}|x_{i})$, respectively.
The first receiver looks for the unique sequence $v_{1}^{n}$ such
that $(v_{1}^{n},y_{1}^{n})$ is jointly typical and declares the
index of the bin containing $v_{1}^{n}$ as the message received. The
second receiver uses the same method to extract the message $w_{2}$.

4) \textit{Error Probability Analysis}: Since the region of
(\ref{l0}) is a subset of Marton's region then, error probability
analysis is the same as \cite{33}.

5) \textit{Equivocation Calculation}: The proof of secrecy
requirement for each individual message (\ref{l1}) and (\ref{l2}) is straightforward and may therefore be omitted.

To prove the requirement of (\ref{l3}) consider $H(W_{1},W_{2}|Z^{n})$, we have
\begin{eqnarray}\nonumber \label{l4}
nR_{e12}&=&H(W_{1},W_{2}|Z^{n})\\ \nonumber  &\geq& H(W_{1},W_{2},Z^{n})-H(Z^{n})\\
\nonumber &=&
H(W_{1},W_{2},V_{1^{n}},V_{2}^{n},Z^{n})-H(V_{1}^{n},V_{2}^{n}|W_{1},W_{2},Z^{n})-H(Z^{n})\\
\nonumber &=&H(W_{1},W_{2},V_{1}^{n},V_{2}^{n})+ H(Z^{n}|W_{1},W_{2},V_{1}^{n},V_{2}^{n})-H(V_{1}^{n},V_{2}^{n}|W_{1},W_{2},Z^{n})-H(Z^{n})\\
\nonumber &\stackrel{(a)}{\geq}&
H(W_{1},W_{2},V_{1}^{n},V_{2}^{n})+H(Z^{n}|W_{1},W_{2},V_{1}^{n},V_{2}^{n})-n\epsilon_{n}-H(Z^{n})\\
\nonumber &\stackrel{(b)}{=}&H(W_{1},W_{2},V_{1}^{n},V_{2}^{n})+H(Z|V_{1}^{n},V_{2}^{n})-n\epsilon_{n}-H(Z^{n}) \\
\nonumber&\stackrel{(c)}{\geq}&
H(V_{1}^{n},V_{2}^{n})+H(Z^{n}|V_{1}^{n},V_{2}^{n}) -n\epsilon_{n}-H(Z^{n}) \\
\nonumber &\stackrel{(d)}{=}&
H(V_{1}^{n})+H(V_{2}^{n})-I(V_{1}^{n};V_{2}^{n})-I(V_{1}^{n},V_{2}^{n};Z^{n})-
n\epsilon_{n}\\
\nonumber &\stackrel{(e)}{=}&
I(V_{1}^{n};Y_{1}^{n})+I(V_{2}^{n};Y_{2}^{n})-I(V_{1}^{n};V_{2}^{n})-I(V_{1}^{n},V_{2}^{n};Z^{n})- n\epsilon_{n}\\
\nonumber &\geq& nR_{1}+nR_{2}-n\epsilon_{n},
\end{eqnarray}

where $(a)$ follows from Fano's inequality which states that for
sufficiently large $n$ we have
$H(V_{1}^{n},V_{2}^{n}|W_{1},W_{2},Z^{n})$ $\leq h(P_{we}^{(n)})$
$+nP_{we}^{n}I(V_{1},V_{2};Z)\leq n\epsilon_{n}$. Here $P_{we}^{n}$
denotes the wiretapper's error probability of decoding
$(v_{1}^{n},v_{2}^{n})$ in the case that the bin numbers $w_{1}$ and
$w_{2}$ are known to the eavesdropper. Since the sum rate is less
than $I(V_{1},V_{2};Z)$, then $P_{we}^{n}\rightarrow 0$ for
sufficiently large $n$. $(b)$ follows from the following Markov
chain: $(W_{1},W_{2})\rightarrow (V_{1},V_{2})\rightarrow$ $ Z$.
Hence, we have
$H(Z^{n}|W_{1},W_{2},V_{1}^{n},V_{2}^{n})=H(Z^{n}|V_{1}^{n},V_{2}^{n})$.
$(c)$ follows from the fact that
$H(W_{1},W_{2},V_{1}^{n},V_{2}^{n})\geq H(V_{1}^{n},V_{2}^{n})$.
$(d)$ follows from that fact that
$H(V_{1}^{n})=I(V_{1}^{n};Y_{1}^{n})$ and
$H(V_{2}^{n})=I(V_{2}^{n};Y_{2}^{n})$.

%Following the same approach, the security conditions for the other
%corner points maybe be proven. Furthermore, by time sharing, we can
%achieve the security conditions for the rate region of (\ref{l0}).
\end{proof}

\section{The secrecy capacity Region of the Degraded BCCs}
In this section, we consider the degraded broadcast channel with
confidential messages and establish its secrecy capacity region.
\begin{defi}
A broadcast channel with confidential messages is said to be
physically degraded, if $X\rightarrow Y_{1} \rightarrow
Y_{2}\rightarrow Z$ forms a Markov chain. In the other words, we
have
\begin{equation}
P(y_{1},y_{2},z|x)=P(y_{1}|x)P(y_{2}|y_{1})P(z|y_{2}).
\end{equation}
\end{defi}
\begin{defi}
A broadcast channel with confidential messages is said to be
stochastically degraded if its conditional marginal distributions
are the same as that of a physically degraded broadcast channel,
i.e., if there exist two distributions $P^{'}(y_{2}|y_{1})$ and
$P^{'}(z|y_{2})$ such that
\begin{eqnarray}
P(y_{2}|x)=\sum_{y_{1}}P(y_{1}|x)P^{'}(y_{2}|y_{1})\\ \nonumber
P(z|x)=\sum_{y_{2}}P(y_{2}|x)P^{'}(z|y_{2})
\end{eqnarray}
\end{defi}
\begin{lem}
The secrecy capacity region of a broadcast channel with confidential
messages depends only on the conditional marginal distributions
$P(y_{1}|x)$, $P(y_{2}|x)$ and $P(z|x)$.
\end{lem}
\begin{proof}
The proof is very similar to \cite{34} and may therefore be omitted
here.
\end{proof}
In the following theorem, we fully characterize the capacity region
of the physically degraded broadcast channel with confidential
messages.
\begin{thm}
The capacity region for transmitting independent secret information
over the degraded broadcast channel is the convex hull of the
closure of all $(R_{1},R_{2})$ satisfying
\begin{IEEEeqnarray}{rl}\label{l7}
    R_{1}&\leq I(X;Y_{1}|U)+I(U;Z)-I(X;Z), \\
    R_{2}&\leq I(U;Y_{2})-I(U;Z).
\end{IEEEeqnarray}
for some joint distribution $P(u)P(x|u)P(y_{1}, y_{2},z|x)$.
\end{thm}
\begin{rem}
If we remove the secrecy constraints by setting
$\mathcal{Z}=\emptyset$, then the above theorem reduces to the
capacity region of the degraded broadcast channel.
\end{rem}
\begin{proof}

\textit{Achievability}: The coding scheme is based on Cover's
superposition coding and the random bining. We refer to this scheme
as Secure Superposition Coding scheme. The available resources at
the encoder are used for two purposes: to confuse the eavesdropper
so that perfect secrecy can be achieved for both layers, and to
transmit the messages in the main channels. To satisfy
confidentiality, the randomization used in the first layer is again
used in the second layer. This makes a shift of $I(U;Z)$ in the
bound of $R_{1}$. The formal proof of the achievability is as
follows:

1) \textit{Codebook Generation}: The structure of the encoder is
depicted in Fig.\ref{f3}.
\begin{figure}
\centerline{\includegraphics[scale=.6]{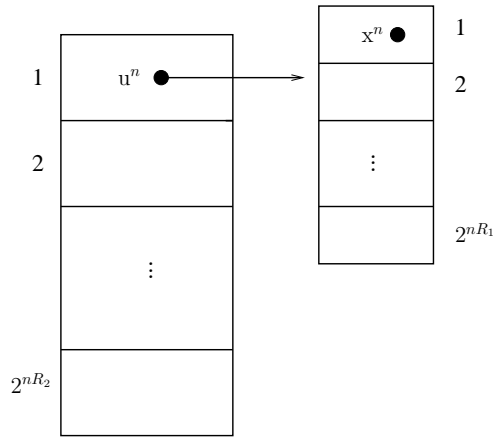}} \caption{Secret
Superposition structure}  \label{f3}
\end{figure}
Let us fix $P(u)$ and $P(x|u)$. The stochastic encoder generates
$2^{n(I(U;Y_{2})-\epsilon)}$ independent and identically distributed
sequences $u^{n}$ according to the distribution
$P(u^{n})=\prod_{i=1}^{n}P(u_{i})$. Next, we randomly distribute
these sequences into $2^{nR_{2}}$ bins such that each bin contains
$2^{n(I(U;Z)-\epsilon)}$ codewords. We index each of the above bins
by $w_{2}\in\{1,2,...,2^{nR_{2}}\}$. For each codeword of $u^{n}$,
it also generates $2^{n(I(X;Y_{1}|U)-\epsilon)}$ independent and
identically distributed sequences $x^{n}$ according to the
distribution $P(x^{n}|u^{n})=\prod_{i=1}^{n}P(x_{i}|u_{i})$. We
randomly distribute these sequences into $2^{nR_{1}}$ bins such that
each bin contains $2^{n(I(X;Z)-I(U;Z)-\epsilon)}$ codewords. We
index each of the above bins by $w_{1}\in\{1,2,...,2^{nR_{1}}\}$.

2) \textit{Encoding}: To send messages $w_{1}$ and $w_{2}$, the
transmitter randomly chooses one of the codewords in bin $w_{2}$,
say $u^{n}$. Then given $u^{n}$, the transmitter randomly chooses
one of $x^{n}$ in bin $w_{1}$ of the second layer and sends it.

3) \textit{Decoding}: The received signal at the legitimate
receivers, $y_{1}^{n}$ and $y_{2}^{n}$, are the outputs of the
channels $P(y_{1}^{n}|x^{n})=\prod_{i=1}^{n}P(y_{1,i}|x_{i})$ and
$P(y_{2}^{n}|x^{n})=\prod_{i=1}^{n}P(y_{2,i}|x_{i})$, respectively.
Receiver $2$ determines the unique $u^{n}$ such that
$(u^{n},y_{2}^{n})$ are jointly typical and declares the index of
the bin containing $u^{n}$ as the message received. If there is none
of such or more than of one such, an error is declared. Receiver $1$
looks for the unique $(u^{n},x^{n})$ such that
$(u^{n},x^{n},y_{1}^{n})$ are jointly typical and declares the
indexes of the bins containing $u^{n}$ and $x^{n}$ as the messages
received. If there is none of such or more than of  one such, an
error is declared.

4) \textit{Error Probability Analysis}: Since each rate pair of
(\ref{l7}) is in the capacity region of the degraded broadcast
channel without confidentiality constraint, then it can be readily
shown that the error probability is arbitrarily small, c.f.
\cite{34}.

5) \textit{Equivocation Calculation}: To prove the secrecy
requirement of (\ref{l1}), we have
\begin{eqnarray}\nonumber
nR_{e1}&=&H(W_{1}|Z^{n})\\ \nonumber
&\geq& H(W_{1}|Z^{n},U^{n})\\ \nonumber &=&
H(W_{1},Z^{n}|U^{n})-H(Z^{n}|U^{n}) \\ \nonumber &=&
H(W_{1},X^{n},Z^{n}|U^{n})-H(Z^{n}|U^{n})-H(X^{n}|W_{1},Z^{n},U^{n})\\
\nonumber
&\stackrel{(a)}{=}&H(W_{1},X^{n}|U^{n})+H(Z^{n}|W_{1},U^{n},X^{n})-H(Z^{n}|U^{n})-n\epsilon_{n}\\
\nonumber &\stackrel{(b)}{\geq}&
H(X^{n}|U^{n})+H(Z^{n}|X^{n})-H(Z^{n}|U^{n})-n\epsilon_{n}\\
\nonumber &\stackrel{(c)}{=}&
H(X^{n};Y_{1}^{n}|U^{n})+I(U^{n};Z^{n})-I(X^{n};Z^{n})-n\epsilon_{n}\\
\nonumber &\geq & nR_{1}-n\epsilon_{n},
\end{eqnarray}
where $(a)$ follows from Fano's inequality which states that
$H(X^{n}|W_{1},Z^{n},U^{n})\leq
h(P_{we}^{(n)})+nP_{we}^{n}I(X;Z)\leq n\epsilon_{n}$ for
sufficiently large $n$. Here $P_{we}^{n}$ denotes the wiretapper's
error probability of decoding $x^{n}$ given the bin number and the
codeword $u^{n}$ are known to the eavesdropper. Since the rate is
less than $I(X;Z)$, then $P_{we}^{n}\rightarrow 0$ for sufficiently
large $n$. $(b)$ follows from the  fact that $(W_{1},U)\rightarrow X
\rightarrow Z$ forms a Markov chain. Thus we have
$I(W_{1},U^{n};Z^{n}|X^{n})=0$, where it is implied that
$H(Z^{n}|W_{1},U^{n},X^{n})=H(Z^{n}|X^{n})$. $(c)$ follows from two
identities: $H(X^{n}|U^{n})=I(X^{n};Y_{1}^{n}|U^{n})$ and
$H(Z^{n}|X^{n})-H(Z^{n}|U^{n})=I(U^{n};Z^{n})-I(X^{n};Z^{n})$. Since the proof of the requirement
(\ref{l2}) is straightforward, we need to prove the requirement of (\ref{l3}).
\begin{eqnarray}\nonumber\label{l6}
nR_{e12}&=&  H(W_{1},W_{2}|Z^{n})\\ \nonumber &\geq& H(W_{1},W_{2},Z^{n})-H(Z^{n})\\
\nonumber &=&
H(W_{1},W_{2},U^{n},X^{n},Z^{n})-H(U^{n},X^{n}|W_{1},W_{2},Z^{n})-H(Z^{n})\\
\nonumber &=&
H(W_{1},W_{2},U^{n},X^{n})\\ \nonumber &+& H(Z^{n}|W_{1},W_{2},U^{n},X^{n})-H(U^{n},X^{n}|W_{1},W_{2},Z^{n})-H(Z^{n})\\
\nonumber &\stackrel{(a)}{\geq}& H(W_{1},W_{2},U^{n},X^{n})+H(Z^{n}|W_{1},W_{2},U^{n},X^{n})-n\epsilon_{n}-H(Z^{n})\\
\nonumber &\stackrel{(b)}{=}&H(W_{1},W_{2},U^{n},X^{n})+H(Z|U^{n},X^{n})-n\epsilon_{n}-H(Z^{n})\\
\nonumber&\stackrel{(c)}{\geq}&
H(U^{n},X^{n})+H(Z^{n}|U^{n},X^{n})-n\epsilon_{n}-H(Z^{n})\\
\nonumber &=&
H(U^{n})+H(X^{n}|U^{n})-I(U^{n},X^{n};Z^{n})- n\epsilon_{n}\\
\nonumber &\stackrel{(d)}{=}&
I(U^{n};Y_{2}^{n})+I(X^{n};Y_{1}^{n}|U^{n})-I(X^{n};Z^{n})-I(U^{n};Z^{n}|X^{n})- n\epsilon_{n}\\
\nonumber &\geq& nR_{1}+nR_{2}-n\epsilon_{n},
\end{eqnarray}
where $(a)$ follows from Fano's inequality that
$H(U^{n},X^{n}|W_{1},W_{2},Z^{n})\leq
h(P_{we}^{(n)})+nP_{we}^{n}I(U,X;Z)\leq n\epsilon_{n}$ for
sufficiently large $n$. Here $P_{we}^{n}$ denotes the wiretapper's
error probability of decoding $(u^{n},x^{n})$ in the case that the
bin numbers $w_{1}$ and $w_{2}$ are known to the eavesdropper. The
eavesdropper first looks for the unique $u^{n}$ in bin $w_{2}$ of
the first layer, such that it is jointly typical with $z^{n}$. Since
the number of candidate codewords is less than $I(U;Z)$, then the
probability of error is arbitrarily small for a sufficiently large
$n$. Next, given $u^{n}$, the eavesdropper looks for the unique
$x^{n}$ in bin $w_{1}$ which is jointly typical with $z^{n}$.
Similarly, since the number of available candidates is less than
$I(X;Z)$, then the probability of error decoding is arbitrarily
small.  $(b)$ follows from the  fact that $(W_{1},W_{2})\rightarrow
U\rightarrow X\rightarrow Z$ forms a Markov chain. Therefore, we
have $I(W_{1},W_{2};Z^{n}|U^{n},X^{n})=0$, where it is implied that
$H(Z^{n}|W_{1},W_{2},U^{n},X^{n})=H(Z^{n}|U^{n},X^{n})$. $(c)$
follows from the fact that $H(W_{1},W_{2},U^{n},X^{n})\geq
H(U^{n},X^{n})$. $(d)$ follows from that fact that
$H(U^{n})=I(U^{n};Y_{2}^{n})$ and
$H(X^{n}|U^{n})=I(X^{n};Y_{1}^{n}|U^{n})$.

 \textit{Converse}: The
transmitter sends two independent secret messages $W_{1}$ and
$W_{2}$ to receiver $1$ and receiver $2$ respectively. Let us define
$U_{i}=(W_{2},Y_{1}^{i-1})$. The following chain of inequality
clarifies the proof:
\begin{eqnarray}\nonumber
nR_{1} &\stackrel{(a)}{\leq}&
\sum_{i=1}^{n}I(W_{1};Y_{1,i}|W_{2},Z_{i},Y_{1}^{i-1},\widetilde{Z}^{i+1})+n\delta_{1}+n\epsilon_{3}\\
\nonumber &=&
 \sum_{i=1}^{n}I(W_{1};Y_{1,i}|U_{i},Z_{i},\widetilde{Z}^{i+1})+n\delta_{1}+n\epsilon_{3}\\
\nonumber &\stackrel{(b)}{\leq}& \sum_{i=1}^{n} I(X_{i};Y_{1,i}|U_{i},Z_{i},\widetilde{Z}^{i+1})+n\delta_{1}+n\epsilon_{3}\\
\nonumber &\stackrel{(c)}{=}&
\sum_{i=1}^{n}I(X_{i};Y_{1,i},U_{i},Z_{i}|\widetilde{Z}^{i+1})-I(X_{i};Z_{i}|\widetilde{Z}^{i+1})-I(X_{i};U_{i}|Z_{i},\widetilde{Z}^{i+1}) +n\delta_{1}+n\epsilon_{3}\\
\nonumber&\stackrel{(d)}{=}& \sum_{i=1}^{n} I(X_{i};Y_{1,i}|U_{i},\widetilde{Z}^{i+1})+I(X_{i};U_{i}|\widetilde{Z}^{i+1})-I(X_{i};Z_{i}|\widetilde{Z}^{i+1})-I(X_{i};U_{i}|Z_{i},\widetilde{Z}^{i+1}) +n\delta_{1}+n\epsilon_{3}\\
\nonumber &\stackrel{(e)}{=}& \sum_{i=1}^{n} I(X_{i};Y_{1,i}|U_{i},\widetilde{Z}^{i+1})-I(X_{i};Z_{i}|\widetilde{Z}^{i+1})+I(Z_{i};U_{i}|\widetilde{Z}^{i+1})-I(Z_{i};U_{i}|X_{i},\widetilde{Z}^{i+1})+n\delta_{1}+n\epsilon_{3}\\
\nonumber &\stackrel{(f)}{=}& \sum_{i=1}^{n}I(X_{i};Y_{1,i}|U_{i},\widetilde{Z}^{i+1})-I(X_{i};Z_{i}|\widetilde{Z}^{i+1})+I(Z_{i};U_{i}|\widetilde{Z}^{i+1})+n\delta_{1}+n\epsilon_{3}\\
\nonumber,
\end{eqnarray}
 $(a)$ follows from the
following lemma (\ref{ll1}). $(b)$ follows
from the data processing theorem. $(c)$ follows from the chain rule.
$(d)$ follows from the fact that
$I(X_{i};Y_{1,i},U_{i},Z_{i}|\widetilde{Z}^{i+1})=I(X_{i};U_{i}|\widetilde{Z}^{i+1})+I(X_{i};Y_{1,i}|U_{i},\widetilde{Z}^{i+1})+I(X_{i};Z_{i}|Y_{1,i},U_{i},\widetilde{Z}^{i+1})$
 and from the fact that $\widetilde{Z}^{i+1}U_{i}\rightarrow X_{i} \rightarrow Y_{1,i} \rightarrow Y_{2,i} \rightarrow
Z_{i}$ forms a Markov chain, which means that
$I(X_{i};Z_{i}|Y_{1,i},U_{i},\widetilde{Z}^{i+1})=0$. $(e)$ follows from the fact that
$I(X_{i};U_{i}|\widetilde{Z}^{i+1})-I(X_{i};U_{i}|Z_{i},\widetilde{Z}^{i+1})=I(Z_{i};U_{i}|\widetilde{Z}^{i+1})-I(Z_{i};U_{i}|X_{i},\widetilde{Z}^{i+1})$.
$(f)$ follows from the fact that $\widetilde{Z}^{i+1} U_{i}\rightarrow
X_{i} \rightarrow Z_{i}$ forms a Markov chain. Thus
$I(Z_{i};U_{i}\widetilde{Z}^{i+1}|X_{i})=0$ which implies that $I(Z_{i};U_{i}|X_{i},\widetilde{Z}^{i+1})=0$.
\begin{lem}: \label{ll1}
For the broadcast channel with confidential messages of
$(W_{1},W_{2})\rightarrow X^{n} \rightarrow
Y_{1}^{n}Y_{2}^{n}Z^{n}$, the perfect secrecy rates are bounded as
follows,
\begin{IEEEeqnarray}{rl}
    nR_{1}&\leq \sum_{i=1}^{n}I(W_{1};Y_{1i}|W_{2},Z_{i},Y_{1}^{i-1},\widetilde{Z}^{i+1})+n\delta_{1}+n\epsilon_{3}, \\
    nR_{2}&\leq \sum_{i=1}^{n}I(W_{2};Y_{2i}|Z_{i},Y_{2}^{i-1},\widetilde{Z}^{i+1})+n\delta_{1}+n\epsilon_{2}.
\end{IEEEeqnarray}
\end{lem}
\begin{proof}
We need to prove the first bound. The second bound can similarly be
proven. According to the above discussion $nR_{1}$ is bounded as
follows:
\begin{eqnarray}\nonumber
nR_{1} &\stackrel{(a)}{\leq}& H(W_{1}|W_{2},Z^{n})+n\epsilon_{3} \\
\nonumber &\stackrel{(b)}{\leq}&
H(W_{1}|W_{2},Z^{n})-H(W_{1}|Y_{1}^{n},W_{2})+n\delta_{1}+n\epsilon_{3}\\
\nonumber &=& I(W_{1};Y_{1}^{n}|W_{2})-I(W_{1};Z^{n}|W_{2})+n\delta_{1}+n\epsilon_{3}
 %&=&H(Y^{n})-H(Z^{n})+H(Z^{n}|W)-H(Y^{n}|W)\\ \nonumber
% &\stackrel{(b)}{=}&\sum_{i=1}^{n}H(Y_{i}|Y^{i-1})-H(Z_{i}|Z^{i-1})+ H(Z_{i}|W)-H(Y_{i}|W)\\ \nonumber
% &\stackrel{(c)}{\leq}&\sum_{i=1}^{n}H(Y_{i}|Y^{i-1})-H(Z_{i}|Z^{i-1},Y^{i-1})+H(Z_{i}|W)-H(Y_{i}|W)\\ \nonumber
% &\stackrel{(d)}{=}&\sum_{i=1}^{n}H(Y_{i}|Y^{i-1})-H(Z_{i}|Y^{i-1})+ H(Z_{i}|W,Y^{i-1})-H(Y_{i}|W,Y^{i-1})\\ \nonumber
% &=&\sum_{i=1}^{n}I(W,Y_{i}|Y^{i-1})-I(W,Z_{i}|Y^{i-1})\\ \nonumber
%&=&\sum_{i=1}^{n}H(W|Z_{i},Y^{i-1})-H(W|Y_{i},Y^{i-1})\\ \nonumber
%&\stackrel{(e)}{=}&\sum_{i=1}^{n}H(W|Z_{i},Y^{i-1})-H(W|Y_{i},Z_{i},Y^{i-1})\\
%&=&\sum_{i=1}^{n}I(W;Y_{i}|Z_{i},Y^{i-1})\nonumber,
\end{eqnarray}
where $(a)$ follows from the secrecy constraint that
$H(W_{1},W_{2}|Z^{n})\geq H(W_{1},W_{2})-n\epsilon_{3}$, the fact
that $H(W_{2}|Z^{n})\leq H(W_{2})$ and the fact that two messages
are independent. $(b)$ follows from Fano's inequality that
$H(W_{1}|Y_{1}^{n},W_{2}) \leq n\delta_{1}$.
%where $(a)$ follows from the Fano's inequality. %$(b)$ follows from the fact that the
%channels are DMC. $(c)$ follows from the fact that conditioning
%decreases entropy. $(d)$ follows from the fact that
%$Z^{i-1}\rightarrow Y^{i-1}\rightarrow Z_{i}$ and
%$Y^{i-1}\rightarrow W \rightarrow Y_{i}\rightarrow Z_{i}$ are Markov
%chains. $(e)$ follows from the fact that, given $Y^{i-1}$,
%$W\rightarrow Y_{i} \rightarrow Z_{i}$ is a Markov chain.
Next, we expand $I(W_{1};Y_{1}^{n}|W_{2})$ and $I(W_{1};Z^{n}|W_{2})$
starting with $I(W_{1};Y_{1}|W_{2})$ and $I(W_{1};\widetilde{Z}^{n}|W_{2})$,
respectively.
\begin{eqnarray}\nonumber
I(W_{1};Y_{1}^{n}|W_{2}) &=&\sum_{i=1}^{n}I(W_{1};Y_{1i}|W_{2},Y_{1}^{i-1}) \\
\nonumber
&=&\sum_{i=1}^{n}I(W_{1},\widetilde{Z}^{i+1};Y_{1i}|W_{2},Y_{1}^{i-1})-I(\widetilde{Z}^{i+1};Y_{1i}|W_{1},W_{2},Y_{1}^{i-1})\\
\nonumber &=&\sum_{i=1}^{n}
I(W_{1};Y_{1i}|W_{2},Y_{1}^{i-1},\widetilde{Z}^{i+1})+I(\widetilde{Z}^{i+1};Y_{1i}|W_{2},Y_{1}^{i-1})-I(\widetilde{Z}^{i+1};Y_{1i}|W_{1},W_{2},Y_{1}^{i-1})\\
\nonumber
&=&\sum_{i=1}^{n}I(W_{1};Y_{1i}|W_{2},Y_{1}^{i-1},\widetilde{Z}^{i+1})+\Delta_{1}-\Delta_{2},
\end{eqnarray}
where,
$\Delta_{1}=\sum_{i=1}^{n}I(\widetilde{Z}^{i+1};Y_{1i}|W_{2},Y_{1}^{i-1})$ and
$\Delta_{2}=\sum_{i=1}^{n}I(\widetilde{Z}^{i+1};Y_{1i}|W_{1},W_{2},Y_{1}^{i-1})$.
Similarly, we have,
\begin{eqnarray}\nonumber
I(W_{1};Z^{n}|W_{2}) &=&\sum_{i=1}^{n}I(W_{1};Z_{i}|W_{2},\widetilde{Z}^{i+1}) \\
\nonumber
&=&\sum_{i=1}^{n}I(W_{1},Y_{1}^{i-1};Z_{i}|W_{2},\widetilde{Z}^{i+1})-I(Y_{1}^{i-1};Z_{i}|W_{1},W_{2},\widetilde{Z}^{i+1})\\
\nonumber &=&\sum_{i=1}^{n}
I(W_{1};Z_{i}|W_{2},Y_{1}^{i-1},\widetilde{Z}^{i+1})+I(Y_{1}^{i-1};Z_{i}|W_{2},\widetilde{Z}^{i+1})-I(Y_{1}^{i-1};Z_{i}|W_{1},W_{2},\widetilde{Z}^{i+1})\\
\nonumber
&=&\sum_{i=1}^{n}I(W_{1};Z_{i}|W_{2},Y_{1}^{i-1},\widetilde{Z}^{i+1})+\Delta_{1}^{*}-\Delta_{2}^{*},
\end{eqnarray}
where,
$\Delta_{1}^{*}=\sum_{i=1}^{n}I(Y_{1}^{i-1};Z_{i}|W_{2},\widetilde{Z}^{i+1})$
and
$\Delta_{2}^{*}=\sum_{i=1}^{n}I(Y_{1}^{i-1};Z_{i}|W_{1},W_{2},\widetilde{Z}^{i+1})$.
According to lemma $7$ of \cite{4}, $\Delta_{1}=\Delta_{1}^{*}$ and
$\Delta_{2}=\Delta_{2}^{*}$. Thus, we have,
\begin{eqnarray}\nonumber
nR_{1}&\leq& \sum_{i=1}^{n}I(W_{1};Y_{1i}|W_{2},Y_{1}^{i-1},\widetilde{Z}^{i+1})-I(W_{1};Z_{i}|W_{2},Y_{1}^{i-1},\widetilde{Z}^{i+1})+n\delta_{1}+n\epsilon_{3}
\\ \nonumber &=&\sum_{i=1}^{n}H(W_{1}|W_{2},Z_{i},Y_{1}^{i-1},\widetilde{Z}^{i+1})-H(W_{1}|W_{2},Y_{1i},Y_{1}^{i-1},\widetilde{Z}^{i+1})+n\delta_{1}+n\epsilon_{3}\\ \nonumber
&\stackrel{(a)}{\leq}&\sum_{i=1}^{n}H(W_{1}|W_{2},Z_{i},Y_{1}^{i-1},\widetilde{Z}^{i+1})-H(W_{1}|W_{2},Y_{1i},Z_{i},Y_{1}^{i-1},\widetilde{Z}^{i+1})+n\delta_{1}+n\epsilon_{3}\\
\nonumber
&=&\sum_{i=1}^{n}I(W_{1};Y_{1i}|W_{2},Z_{i},Y_{1}^{i-1},\widetilde{Z}^{i+1})+n\delta_{1}\nonumber+n\epsilon_{3},
\end{eqnarray}
where $(a)$ follows from the fact that conditioning always decreases
the entropy.
\end{proof}
For the second receiver, we have
\begin{eqnarray}\nonumber
nR_{2} &\stackrel{(a)}{\leq}&\sum_{i=1}^{n}I(W_{2};Y_{2,i}|Y_{2}^{i-1},Z_{i},\widetilde{Z}^{i+1})+n\delta_{2}+n\epsilon_{1}\\
\nonumber
&=&\sum_{i=1}^{n}H(Y_{2,i}|Y_{2}^{i-1},Z_{i},\widetilde{Z}^{i+1})-H(Y_{2,i}|W_{2},Y_{2}^{i-1},Z_{i},\widetilde{Z}^{i+1})+n\delta_{2}+n\epsilon_{1}\\
\nonumber &\stackrel{(b)}{\leq}&\sum_{i=1}^{n}H(Y_{2,i}|Z_{i},\widetilde{Z}^{i+1})-H(Y_{2,i}|W_{2},Y_{1}^{i-1},Y_{2}^{i-1},Z_{i},\widetilde{Z}^{i+1})+n\delta_{2}+n\epsilon_{1}\\
\nonumber &\stackrel{(c)}{=}&\sum_{i=1}^{n}H(Y_{2,i}|Z_{i},\widetilde{Z}^{i+1})-H(Y_{2,i}|U_{i},Z_{i},\widetilde{Z}^{i+1})+n\delta_{2}+n\epsilon_{1}\\
\nonumber &=&\sum_{i=1}^{n}I(Y_{2,i};U_{i}|Z_{i},\widetilde{Z}^{i+1})+n\delta_{2}+n\epsilon_{1}\\
\nonumber &=&\sum_{i=1}^{n}I(Y_{2,i};U_{i}|\widetilde{Z}^{i+1})+I(Y_{2,i};Z_{i}|U_{i},\widetilde{Z}^{i+1})-I(Y_{2,i};Z_{i}|\widetilde{Z}^{i+1})+n\delta_{2}+n\epsilon_{1}\\
\nonumber &=&\sum_{i=1}^{n}I(Y_{2,i};U_{i}|\widetilde{Z}^{i+1})-I(Z_{i};U_{i}|\widetilde{Z}^{i+1})+I(Z_{i};U_{i}|Y_{2,i},\widetilde{Z}^{i+1})+n\delta_{2}+n\epsilon_{1}\\
\nonumber
&\stackrel{(d)}{=}&\sum_{i=1}^{n}I(Y_{2,i};U_{i}|\widetilde{Z}^{i+1})-I(Z_{i};U_{i}|\widetilde{Z}^{i+1})+n\delta_{2}+n\epsilon_{1}\nonumber,
\end{eqnarray}
where $(a)$ follows from the lemma (\ref{ll1}). $(b)$
follows from the fact that conditioning always decreases the entropy.
$(c)$ follows from the fact that $Y_{2}^{i-1}\rightarrow W_{2}\widetilde{Z}^{i+1}Y_{1}^{i-1}\rightarrow Y_{2i}\rightarrow Z_{i}$ forms a Markov chain. $(d)$ follows from the fact that
$\widetilde{Z}^{i+1}U_{i}\rightarrow Y_{2,i} \rightarrow Z_{i}$ forms a
Markov chain. Thus $I(Z_{i};U_{i}\widetilde{Z}^{i+1}|Y_{2i})=0$ which implies that $I(Z_{i};U_{i}|Y_{2i},\widetilde{Z}^{i+1})=0$. Now, following \cite{34}, let us define the time
sharing random variable $Q$ which is uniformly distributed over
$\{1,2,...,n\}$ and independent of
$(W_{1},W_{2},X^{n},Y_{1}^{n},Y_{2}^{n})$. Let us define
$U=U_{Q},~ V=(\widetilde{Z}^{Q+1},Q),~ X=X_{Q},~ Y_{1}=Y_{1,Q},~ Y_{2}=Y_{2,Q},~ Z=Z_{Q}$,
then we can bound $R_{1}$ and $R_{2}$ as follows
\begin{IEEEeqnarray}{rl}
    R_{1}&\leq I(X;Y_{1}|U,V)+I(U;Z|V)-I(X;Z|V), \\
    R_{2}&\leq I(U;Y_{2}|V)-I(U;Z|V).
\end{IEEEeqnarray}
Since Conditional mutual informations are average of unconditional
ones, the maximum region is achieved when $V$ is a constant. This
proves the converse part.
\end{proof}
\section{Gaussian BCCs}
In this section we consider the physically degraded AWGN broadcast channel with confidential messages. We show that secret superposition coding with
 Gaussian codebook is optimal. At time $i$ the received signals are $Y_{1i}=X_{i}+n_{1i}$,  $Y_{2i}=X_{i}+n_{2i}$ and  $Z_{i}=X_{i}+n_{3i}$, where $n_{1i}$'s,
$n_{2i}$'s and $n_{3i}$'s are each independent identically distributed Gaussian random variables with zero means and $Var(n_{ji})=N_{j}$, j=1,2,3. All noises are independent
of $X_{i}$ and $N_{1}\leq N_{2}\leq N_{3}$. Assume that transmitted power is limited to $E[X^{2}]\leq P$. Since the channels are degraded, at time $i$,
$Y_{1i}=X_{i}+n_{1i}$,  $Y_{2i}=Y_{1i}+n_{2i}^{'}$ and  $Z_{i}=Y_{2i}+n_{3i}^{'}$, where $n_{1i}$'s are i.i.d $\mathcal{N}(0,N_{1})$,
$n_{2i}^{'}$'s are i.i.d $\mathcal{N}(0,N_{1}-N_{2})$, and $n_{3i}^{'}$'s are i.i.d $\mathcal{N}(0,N_{2}-N_{3})$. Fig.\ref{f4} shows the equivalent channels for the physically
degraded AWGN-BCCs. The following theorem illustrates the secrecy capacity region of AWGN-BCCs.
\begin{thm}
The secrecy capacity region of the AWGN broadcast channel with confidential messages is given by the set of rates pairs $(R_{1},R_{2})$
such that
\begin{IEEEeqnarray}{rl}
    R_{1}&\leq C\left(\frac{\alpha P}{N_{1}}\right)+C\left(\frac{(1-\alpha) P}{\alpha P +N_{3}}\right)-C\left(\frac{P}{N_{3}}\right), \\
    R_{2}&\leq C\left(\frac{(1-\alpha) P}{\alpha P +N_{2}}\right)-C\left(\frac{(1-\alpha) P}{\alpha P +N_{3}}\right).
\end{IEEEeqnarray}
for some $\alpha \in [0,1]$.
\end{thm}
\begin{figure}
\centerline{\includegraphics[scale=.7]{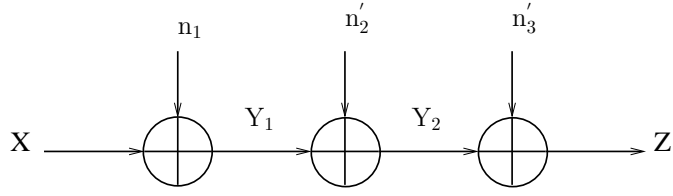}} \caption{equivalent channels for the AWGN-BCCs}  \label{f4}
\end{figure}
\begin{proof}

\textit{Achievability}: Let $U\sim \mathcal{N}(0,(1-\alpha)P)$ and $X^{'}\sim \mathcal{N}(0,\alpha P)$ be independent and $X=U+X^{'} \sim \mathcal{N}(0,P)$.
Therefore, the amount of $I(X;Y_{1}|U)$, $I(U;Z)$, $I(X;Z)$, and $I(U;Y_{2})$ can be easily evaluated.
Now consider the following secure superposition coding scheme:

1) \textit{Codebook Generation}: Generate $2^{nI(U;Y_{2})}$ i.i.d Gaussian codewords $u^{n}$ with average power $(1-\alpha)P$
and randomly distribute these codewords into $2^{nR_{2}}$ bins. Then index each bin by $w_{2}\in\{1,2,...,2^{nR_{2}}\}$.
Generate an independent set of $2^{nI(X^{'};Y_{1}})$ i.i.d Gaussian codewords $x^{'n}$ with average power $\alpha P$. Then, Randomly distribute
them into $2^{nR_{1}}$ bins. Index each bin by $w_{1}\in\{1,2,...,2^{nR_{1}}\}$.

2) \textit{Encoding}: To send messages $w_{1}$ and $w_{2}$, the
transmitter randomly chooses one of the codewords in bin $w_{2}$,
(say $u^{n}$) and one of the codewords in bin $w_{1}$ (say $x^{'n}$ ). Then, simply transmits
$x^{n} =u^{n}+ x^{'n}$.

3) \textit{Decoding}: The received signal at the legitimate
receivers are $y_{1}^{n}$ and $y_{2}^{n}$ respectively.
Receiver $2$ determines the unique $u^{n}$ such that
$(u^{n},y_{2}^{n})$ are jointly typical and declares the index of
the bin containing $u^{n}$ as the message received. If there is none
of such or more than of one such, an error is declared. Receiver $1$
uses successive cancelation method; first decodes $u^{n}$ and subtracts off $y_{1}^{n}$ and then
looks for the unique $x^{'n}$ such that
$(x^{'n},y_{1}^{n})$ are jointly typical and declares the
indexes of the bin containing $x^{'n}$ as the message
received.

The error probability analysis and equivocation calculation is straightforward and may therefor be omitted.

\textit{Converse}: According to the previous section, $R_{2}$ is bounded as follows:

\begin{eqnarray}\label{l10}
R_{2} \leq I(Y_{2};U|Z) = H(Y_{2}|Z)-H(Y_{2}|U,Z)
\end{eqnarray}
The classical entropy power inequality states that:
\begin{eqnarray}\nonumber
2^{\frac{2}{n}H(Y_{2}+n_{3}^{'})}\geq 2^{\frac{2}{n}H(Y_{2})}+2^{\frac{2}{n}H(n_{3}^{'})}
\end{eqnarray}
Therefore, $H(Y_{2}|Z)$ can be written as follows:
\begin{eqnarray}\nonumber
H(Y_{2}|Z)&=&H(Z|Y_{2})+H(Y_{2})-H(Z)\\ \nonumber
&=& \frac{n}{2}\log(N_{3}-N_{2})+H(Y_{2})-H(Y_{2}+n_{3}^{'})\\ \nonumber
&\leq&\frac{n}{2}\log(N_{3}-N_{2})+H(Y_{2})-\frac{n}{2}\log(2^{\frac{2}{n}H(Y_{2})}+N_{3}-N_{2})
\end{eqnarray}
On the other hand, for any fixed $a\in \mathcal{R}$, the function
\begin{eqnarray}\nonumber
f(t,a)=t-\frac{n}{2}\log(2^{\frac{2}{n}t}+a)
\end{eqnarray}
is concave in $t$ and has a global maximum at $t=t_{max}$. Thus, $H(Y_{2}|Z)$ is maximized when $Y_{2}$ (or equivalently $X$) has Gaussian distribution. Hence,
\begin{eqnarray} \label{l8}
H(Y_{2}|Z)&\leq& \frac{n}{2}\log(N_{3}-N_{2})+\frac{n}{2}\log(P+N_{2})-\frac{n}{2}\log(P+N_{3})\\ \nonumber
&=&\frac{n}{2}\log\left(\frac{(N_{3}-N_{2})(P+N_{2})}{P+N_{3}}\right)
\end{eqnarray}
Now consider the term $H(Y_{2}|U,Z)$. This term is lower bounded with $H(Y_{2}|U,X,Z)=\frac{n}{2}\log(N_{2})$ which is greater
 than $\frac{n}{2}\log(\frac{N_{2}(N_{3}-N_{2})}{N_{3}})$.
Hence,
\begin{eqnarray}\label{l9}
\frac{n}{2}\log(\frac{N_{2}(N_{3}-N_{2})}{N_{3}})\leq H(Y_{2}|U,Z)\leq H(Y_{2}|Z)
\end{eqnarray}
Inequalities (\ref{l8}) and (\ref{l9}) imply that there exists a $\alpha \in [0,1]$ such that
\begin{eqnarray}\label{l11}
H(Y_{2}|U,Z)=\frac{n}{2}\log\left(\frac{(N_{3}-N_{2})(\alpha P+N_{2})}{\alpha P+N_{3}}\right)
\end{eqnarray}
Substituting (\ref{l11}) and (\ref{l8}) into (\ref{l10}) yields the desired bound
\begin{eqnarray}
R_{2} &\leq&  H(Y_{2}|Z)-H(Y_{2}|U,Z) \\ \nonumber
&\leq&\frac{n}{2}\log\left(\frac{(P+N_{2})(\alpha P + N_{3})}{(P+N_{3})(\alpha P + N_{2})}\right)\\ \nonumber
&=& C\left(\frac{(1-\alpha) P}{\alpha P +N_{2}}\right)-C\left(\frac{(1-\alpha) P}{\alpha P +N_{3}}\right)
\end{eqnarray}
To bound the rate $R_{1}$, we need the following generalized entropy power inequality which is proven in \cite{35}.
\begin{lem} \cite{35}:
Let $n_{1}$, $n_{2}$ be two gaussian random variables. Let $U$ be a random variable independent of $n_{1}$ and $n_{2}$. Consider
The optimization problem
\begin{eqnarray}
\max_{P(X|U)}&& H(X+n_{1}|U)-H(X+n_{2}|U)\\
\text{subject to}&&  Var(X|U)\leq s
\end{eqnarray}
where the maximization is over all distribution of $X$ given $U$ independent of $n_{1}$ and $n_{2}$. A Gaussian $P(x|u)$ with same variance for each $u$ is an optimal solution
for this optimization problem.
\end{lem}
The rate $R_{1}$ is bounded as follows
\begin{eqnarray}
R_{1}&\leq& I(X;Y_{1}|U)-I(X;Z)+I(U;Z)\\ \nonumber
&=&H(Y_{1}|U)-H(Y_{1}|X,U)+H(Z|X)-H(Z|U) \\ \nonumber
&=&H(Y_{1}|U)-H(Z|U)+\frac{n}{2}\log(\frac{N_{3}}{N_{1}})\\ \nonumber
&=&H(X+n_{1}|U)-H(X+n_{3}|U)+\frac{n}{2}\log(\frac{N_{3}}{N_{1}})
\end{eqnarray}
On the other hand using (\ref{l11}), when $Z=0$ and $n_{2}=0$ then $Var(X|U)=\alpha P$. Therefore, According to the above lemma the Gaussian distribution is
optimum and $R_{1}$ is bounded as
\begin{eqnarray}
R_{1}&\leq& \frac{n}{2}\log\left( \frac{\alpha P + N_{1}}{\alpha P + N_{3}}\frac{N_{3}}{N_{1}} \right)\\ \nonumber
&=& C\left(\frac{\alpha P}{N_{1}}\right)+C\left(\frac{(1-\alpha) P}{\alpha P +N_{3}}\right)-C\left(\frac{P}{N_{3}}\right)
\end{eqnarray}
\end{proof}
\section{Conclusion}
A generalization of the wire-tap channel to the case of two
receivers and one eavesdropper is considered. We established an
inner bound for the general (non-degraded) case. This bound matches
Marton's bound on broadcast channels without security constraint.
Furthermore, we considered the scenario in which the channels are
degraded. We established the perfect secrecy capacity region for
this case. The achievability coding scheme is a secret superposition
scheme where randomization in the first layer helps the secrecy of
the second layer. The converse proof combines the converse proof for
the degraded broadcast channel without security constraint and the
perfect secrecy constraint. We proved that the secret superposition
scheme with Gaussian codebook is optimal in AWGN-BCCs. The converse
proof is based on the generalized entropy power inequality and
Csiszar lemma.

\end{document}